\documentclass[12pt,a4paper,thmsb]{article}
\usepackage{amsmath}
\usepackage{amssymb}
\usepackage{amsfonts}
\usepackage{color}
\usepackage{enumerate}
\newtheorem{theorem}{Theorem}

\newtheorem{proposition}{Proposition}
\newtheorem{corollary}{Corollary}

\newenvironment{proof}[1][Proof]{\noindent\textbf{#1.} }{\ \rule{0.5em}{0.5em}}

\newcommand*\samethanks[1][\value{footnote}]{\footnotemark[#1]}
\begin{document}

\title{\textbf{Sharing the proceeds from a hierarchical venture when agents have needs}\thanks{R. Pablo Arribillaga and Pablo Neme acknowledge financial support from Universidad Nacional de San Luis (UNSL) through grants 032016, 031620 and 031323, from Consejo Nacional de Investigaciones Científicas y Técnicas (CONICET) through grant PIP 112-200801-00655, and from Agencia Nacional de Promoción Científica y Tecnológica through grant PICT 2017-2355. Juan D. Moreno-Ternero acknowledges grant PID2023-146364NB-I00, funded by MCIU/AEI/10.13039/501100011033 and FSE+.
}}

\author{\textbf{R. Pablo Arribillaga}%
\thanks{Instituto de Matem\'{a}tica Aplicada San Luis, Universidad Nacional de San Luis and CONICET. 
Ej\'{e}rcito de los Andes 950. 
5700, San Luis, Argentina; rarribi@unsl.edu.ar (Arribillaga) and paneme@unsl.edu.ar (Neme).%
} \\
\textbf{Juan D. Moreno-Ternero}%
\thanks{Department of Economics, Universidad Pablo de Olavide, Spain; email: jdmoreno@upo.es%
}
\\ 
\textbf{Pablo Neme}\samethanks [2]
\\ 
}
\maketitle

\begin{abstract}
We consider a setting in which a set of agents are hierarchically organized for a joint venture. They each generate revenues for the joint venture and have individual needs to cover. The aim is to distribute aggregate revenues appropriately. We characterize a family of need-adjusted geometric rules where the net revenue (after covering needs) `bubbles up' in the hierarchy, as well as a need-adjusted serial rule in which the net revenue is equally shared among each agent and his predecessors in the hierarchy. 
\end{abstract}
\medskip{}

\medskip{}

\medskip{}

\newpage

\section{Introduction}
Agents typically engage in joint ventures. Often, these ventures are hierarchically organized, with each agent located in a different layer of the hierarchy, reflecting a different degree of responsibility, command, or even seniority (e.g., Garicano, 2000; Demange, 2004). In this paper, we are concerned with a basic resource allocation problem arising in those settings. That is, suppose agents generate revenues to the (hierarchical) venture, but also have needs to cover. The issue is how to allocate the overall revenue among participating agents, while taking into account their needs and locations in the hierarchy. 
Thus, we follow Hougaard et al. (2017), who introduced the problem of sharing revenues in a hierarchy.\footnote{Related problems had been studied much earlier in the literature (e.g., Claus and Kleitman, 1973; Littlechild and Owen, 1973; Bird, 1976). More recently, see Hougaard et al. (2022) and Gudmundsson et al. (2025).} The novelty of our model is to introduce individual needs in the problem.\footnote{Similarly, Mart\'{\i}nez and Moreno-Ternero (2024, 2026a) have recently extended the basic setting of redistribution problems (e.g., Ju et al., 2007) to account for needs. 
} The requirement that the allocation process grants each individual an amount at least as large as to cover the need will be crucial for our analysis. We shall formalize this principle as the axiom of \textit{needs lower bound}, which aligns with the long tradition in axiomatic work that endorses lower bounds as a basic axiom of distributive justice (e.g., Rawls, 1971; Moulin, 2004; Moreno-Ternero and Roemer, 2006).

Hougaard et al. (2017) characterize in their setting a family of geometric rules in which revenue `bubbles up' in the hierarchy. With a geometric rule, the lowest-ranked agent gets a share of his revenue, his immediate superior gets the same share of his resulting revenue, after adding the remaining `surplus' from the lowest-ranked agent, etc. This is reminiscent of the so-called geometric incentive tree mechanisms, which are proven successful for social mobilization (e.g., Pickard et al., 2011). We characterize in this paper \textit{need-adjusted geometric rules}, which extend the family of geometric rules mentioned above to the case in which agents have needs (which are measured in the same monetary units as the revenues). In our benchmark case, we assume that each agent in the hierarchy has a need below the revenue he brings to the hierarchy. With a need-adjusted geometric rule, the lowest-ranked agent gets his need, as well as a share of his remaining revenue, his immediate superior also gets his need, as well as the same share of his resulting revenue, after adding the remaining `surplus' from the lowest-ranked agent, etc. We show that the family of need-adjusted geometric rules is characterized by the needs lower bound axiom, together with the same axioms Hougaard et al. (2017) considered in their characterization (with the exception of a specific modification of \textit{scale invariance}). 

We highlight three focal rules within the family, arising when the share of (net) revenue agents keep is, respectively, zero, one-half and one. That is, the \textit{need-adjusted full-transfer rule}, the \textit{need-adjusted balanced transfer rule} and the \textit{no-transfer rule}, respectively. We characterize each of those rules adding one additional axiom to the list characterizing the family of geometric rules. 

We then single out another rule outside the family of need-adjusted geometric rules. This rule states that each agent keeps his need and shares the residual equally with all his predecessors. The rule is dubbed the \textit{serial rule} as it is reminiscent of the cost-sharing rule introduced by Moulin and Shenker (1992). We characterize this rule by the combination of one of the axioms in the characterization of the need-adjusted geometric rules, and two axioms strengthening two other axioms used therein. 

Our results also have connections to the literature on river sharing. 
Ambec and Sprumont (2002) is the seminal contribution to analyze the allocation of water that flows along a river among the agents located therein. 
More recently, Mart\'{\i}nez and Moreno-Ternero (2025a, 2026b) have analyzed a resourcist version of this model of river sharing in which agents are not endowed with utility functions. As a consequence, the model becomes a mirror version of the problem of revenue sharing in hierarchies, with no needs. They characterize counterpart rules for the geometric and serial rules mentioned above, without adjusting for needs. Earlier, Ansink and Weikard (2012, 2015) extended the river sharing model to account for agents' claims on the river water too, which can be seen as similar to the concept of needs in our model. Finally, Yang et al. (2025) and Mart\'{\i}nez and Moreno-Ternero (2025b) have studied another stylized model to analyze the problem of distributing a budget of emissions permits among agents located along a river.\footnote{Ni and Wang (2007) and Dong et al. (2012) studied earlier another model of sharing the cost of cleaning a polluted river.} 

The rest of the paper is organized as follows. In Section \ref{seccion the model}, we introduce the model and basic axioms for our analysis. In Section 3, we present the main results of our analysis. First, the characterization of the family of need-adjusted geometric rules. Then, the independent characterizations of each of the three focal elements within the family. Finally, the characterization of the need-adjusted serial rule. In Section 4, we explore three extensions of our benchmark analysis. First, we consider the larger domain of problems in which needs can be collectively (but not necessarily individually) covered by revenues. Second, we consider the restricted domain in which needs are all null. Finally, we study in more detail the canonical two-agent case.
\newpage

\section{The model}\label{seccion the model}
Suppose there exists a set of potential
\textbf{agents}, identified with the set of natural numbers. Let
$\mathcal{M}$ be the class of finite subsets of the natural numbers,
with generic element $M$. Each set $M\in \mathcal{M}$ will represent
a \textbf{linear hierarchy}, with the convention that lower numbers
in $M$ refer to lower positions in the hierarchy. For instance, if
$M=\{1,\dots, m\}$, then $1$ is representing the agent with the
lowest rank in the hierarchy, whereas $m$ is representing the agent
with the highest rank. 

Agents in each (linear) hierarchy will be involved in a joint venture
to which all of them contribute. Formally, for each $i\in M$, let
$r_{i}\in\mathbb{R}_{+}$ be the \textbf{revenue} that agent $i$
generates, and $r=(r_{i})_{i\in M}$ the profile of
revenues. Agents also have individual needs, assumed to be below the revenues they generate. Formally, for each $i\in M$, let
$z_i\le r_{i}\in\mathbb{R}_{+}$ be the \textbf{need} that agent $i$
has, and $z=(z_{i})_{i\in M}$ the profile of
needs. 
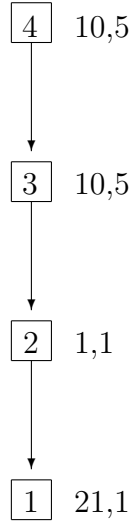
A {\it linear hierarchy with needs revenue sharing problem}, or simply, a
\textbf{problem} is a triplet consisting of a (linear) hierarchy
$M\in\mathcal{M}$, a profile of revenues
$r\in\mathbb{R}_{+}^{|M|}$ and a profile of needs
$z\in\mathbb{R}_{+}^{|M|}$ such that $r_i\ge z_{i}$, for each $i\in M$. Let $\mathcal{Z}^{M}$ be the set of
problems involving the hierarchy $M$ and
$\mathcal{Z}=\bigcup_{M\in\mathcal{M}}\mathcal{Z}^{M}$.

Given a problem $(M,r,z)\in \mathcal{Z}$, an \textbf{allocation} is a
vector $x\in\mathbb{R}^{|M|}_+$ satisfying {\bf balance}, i.e.,
$\sum_{i\in M}x_{i}=\sum_{i\in M}r_{i}$.
An \textbf{allocation rule} is a mapping $\phi$ assigning to each
problem $(M,r,z)\in \mathcal{Z}$ an allocation $\phi(M,r,z)$. We assume
from the outset that rules are \textbf{anonymous}, i.e., for each
problem $(M,r,z)\in \mathcal{Z}$, and for each strictly monotonic function $g:M\to M'$, $\phi_{g(i)}(M',r',z')=\phi_{i}(M,r,z)$,
where $(r'_{g(i)},z'_{g(i)})=(r_{i},z_{i})$, for each $i\in M$. Thus, in what follows for this section, we assume, without loss of generality, that
$M=\{1,\dots,m\}.$

We now consider several axioms for allocation rules. First, the axiom indicating that all agents are guaranteed at least their individual needs. Formally,

\bigskip\noindent \textbf{Needs Lower Bound}: For each $(M,r,z)\in \mathcal{Z}$, and each $i\in M$,
\[
\phi_{i}(M,r,z)\ge z_i.
\]

We now consider a minimalistic version of the principle of consistency (e.g., Thomson, 2012). 
Suppose the agent with the lowest rank leaves the hierarchy, after the allocation takes place, and the immediate predecessor receives the residual revenue. The axiom
then states that the solution of the new problem agrees with the
solution of the original problem for all the standing agents in the
hierarchy.\footnote{
The axiom extends the one introduced by Hougaard et al., (2017) for the case without needs, which is itself reminiscent of the so-called ``first agent consistency'' axiom
proposed by Potters and Sudh\"olter (1999) for airport problems.}
Formally,

\bigskip\noindent \textbf{Lowest Rank Consistency}: For $|M|\geq 2,$ where $(M,r,z)\in
\mathcal{Z}$, and
$(M\setminus\{1\},(r_{2}+r_{1}-\phi_{1}(M,r,z),r_{M\setminus\{1,2\}}),z_{M\setminus\{1\}})
\in \mathcal{Z},$ we have,
\[
\phi_{M\setminus\{1\}
}(M,r,z)=\phi\left(M\setminus\{1\},(r_{2}+r_{1}-\phi_{1}(M,r,z),r_{M\setminus\{1,2\}}),z_{M\setminus\{1\}}\right).
\]

We now consider axioms referring instead to the top of the hierarchy. First, the axiom stating that the characteristics of the top agent are irrelevant for the allocation process of the remaining agents. Formally, 

\bigskip\noindent \textbf{Highest Rank Independence}: For each $(M,r,z)\in \mathcal{Z}$, each $\hat{r}_{m}\in \Bbb{R}_{+}$ and each $\hat{z}_{m}\in \Bbb{R}_{+}$ such that $\hat{r}_{m}\ge\hat{z}_{m}$, 
\[
\phi_{M\setminus\{m\}}(M,r,z)=\phi_{M\setminus\{m\}}\left(M,(r_{-m},\hat{r}_{m}),(z_{-m},\hat{z}_{m})\right).
\]

The next axiom refers to hypothetical manipulations from the top agent, who might be bringing a follower to the hierarchy to reallocate revenues and needs. The axiom states that such a move should not affect the remaining agents in the hierarchy.\footnote{The axiom extends the namesake axiom introduced by Hougaard et al., (2017) for the case without needs, which is itself a specific version of the merging-splitting proofness axioms studied in Ju et al. (2007) and Ju (2013).} Formally,

\bigskip\noindent
\textbf{Highest Rank Splitting Neutrality}: For each $(M,r,z)\in
\mathcal{Z}$, let $(M',r',z')\in \mathcal{Z}$ be such that $M'=M\cup
\{k\}$, $k> m$, $r_{m}=r'_{k}+r'_{m}$,   $z_{m}=z'_{k}+z'_{m}$, $r_{M\backslash
\{m\}}^{\prime }=r_{M\backslash \{m\}}$, and $z_{M\backslash
\{m\}}^{\prime }=z_{M\backslash \{m\}}$. Then, $$ \phi_{M\backslash
\{m\}}(M',r',z')=\phi_{M\backslash \{m\}}\left( M,r,z\right). $$

We conclude with an axiom applying only to two-agent problems. It formalizes the standard notion of linearity (over agents' characteristics).

\bigskip\noindent
 \textbf{Bilateral Linearity}: For each pair $(M,r,z),(M,r',z')\in \mathcal{Z}$ with $M=\{1,2\}$, and each pair $\alpha,\beta>0$,
 \[
\phi(M,\alpha r+\beta r',\alpha z+\beta z')\ =\alpha\phi(M,r,z)+\beta\phi(M,r',z').
\]

\section{The main results}

\subsection{Need-adjusted geometric rules}
We now introduce a family of rules, formalizing the procedure that the lowest-ranked agent gets his need, as well as a share $\lambda\in[0,1]$ of
her revenue, his immediate superior gets his need, a share $\lambda$ of his
augmented revenue (after accounting for any `surplus' from the lowest-ranked
agent), etc., and the highest-ranked agent gets the residual. Formally,
if $M=\{1,\dots, m\}$, payment shares are determined recursively as

$x_{1}^{\lambda}=z_{1}+\lambda(r_{1}-z_{1})$,

$x_{2}^{\lambda}=z_{2}+\lambda(r_{2}-z_{2}+(r_{1}-x_{1}^{\lambda}))=z_{2}+\lambda(r_{2}-z_{2}+(1-\lambda)(r_{1}-z_{1}))$,


and so forth. That is,
\begin{equation}
x_{i}^{\lambda}=z_{i}+\lambda \left(r_{i}-z_i\right)+(1-\lambda)x_{i-1}^{\lambda}
\label{lambdaz}
\end{equation}
for each $i\in M\setminus\{m\}$, with the notational convention that $x_{0}^{\lambda}=0$. Furthermore,
\begin{equation}
x_{m}^{\lambda}=\sum_{i=1}^{m}r_{i}-\sum_{i=1}^{m-1}x_{i}^{\lambda}.\label{lambdazm}
\end{equation}
Note that (\ref{lambdaz}) and (\ref{lambdazm}) can be given the closed-form
expressions
$$x_i^{\lambda} = z_{i}+\lambda \left(r_i-z_{i} + (1- \lambda)(r_{i-1}-z_{i-1}) + \dots + (1-\lambda)^{i-1}(r_1-z_1)\right),$$ for $i= 1,\dots,
m-1$ and
$$x_m^{\lambda} = r_m + (1- \lambda)(r_{m-1}-z_{m-1}) + \dots + (1-\lambda)^{m-1}(r_1-z_1).$$

Denote the corresponding family of rules so defined,
which we call \textbf{need-adjusted geometric rules}, by
$\{\phi^{\lambda}\}_{\lambda\in[0,1]}$.

\bigskip
\noindent \textbf{Example 1:} Consider the problem
$(\{1,2,3,4\},(21,1,10,10),(1,1,5,5))$, i.e., the (linear) hierarchy made of four
agents, $1$, $2$, $3$, and $4$, in which agent $1$ generates a revenue of
$21$ and has a need of $1$, agent $2$ a revenue of $1$ and has a need of $1$, agent $3$ a revenue of $10$, and has a need of $5$, and agent $4$ a revenue of $10$, and has a need of $5$. 
Figure 1 below illustrates the situation.

The need-adjusted geometric rules select the allocation
$$
(1+20\lambda,1+20\lambda(1-\lambda),5(1+\lambda)-20(1-\lambda)\lambda^2),
$$
for each $\lambda\in[0,1]$. 
\newpage

$$
 \begin{picture}(190,40)
 \put(100,10){\line(0,-1){40}}
 \put(92.5,10){\dashbox{15}(15,15)}
 \put(100,-50){\line(0,-1){40}}
 \put(92.5,-50){\dashbox{15}(15,15)}
 \put(100,-30){\vector(0,-1){1}}
 \put(100,-90){\vector(0,-1){1}}
\put(92.5,-110){\dashbox{15}(15,15)}
\put(100,-110){\line(0,-1){40}}
\put(100,-150){\vector(0,-1){1}}
\put(92.5,-170){\dashbox{15}(15,15)}
\put(96.5,12){4}\put(116,-167){21,1}\put(97,-167){1}
\put(96.5,-47){3}\put(116,-107){1,1}\put(97,-107){2}
\put(116,-47){10,5}\put(116,12){10,5}
 \thicklines
\end{picture}
$$

\bigskip

\bigskip
\vspace{5.5 cm} \qquad\qquad\qquad\qquad\qquad\quad\footnotesize{\textbf{Figure 1: A (linear) hierarchy with needs}.}

\normalsize
\bigskip

As the next result states, the above family is characterized by the combination of the axioms introduced above. 
\begin{theorem}\label{teorema de caracterizacion de la geometrica con r>z}
    A rule satisfies needs lower bound, lowest rank consistency, highest rank independence, highest rank splitting neutrality and bilateral linearity if and only if it is a need-adjusted geometric rule.
\end{theorem}
\textit{Proof:} It is not difficult to see that the need-adjusted geometric rules
satisfy all the axioms in the statement of the theorem. 
Conversely, let $\phi$ be a rule satisfying all the axioms in
the statement of the theorem. 
The proof follows by mathematical induction:
\begin{description}
    \item[\textbf{Base step:} $\boldsymbol{M=\{1,2\}}.$] Let $r=(r_{1},r_{2})$ and $z=(z_{1},z_{2})$. 
By \textit{needs lower bound}, $\phi_{1}(M,r,z)\geq z_1$.
Now we claim that $\phi_{1}(M,r,z)\leq r_{1}$, so $\phi_{1}(M,r,z)=z_1+\lambda
(r_{1}-z_1)=\phi_{1}^{\lambda}(M,r,z)$ for some $\lambda\in[0,1]$. Indeed,
assume that $\phi_1(M,r,z) > r_1;$ then by \textit{highest rank independence} $\phi_1(M,r,z) = \phi_1(M, (r_1,0),(z_1,0))$ so by balance
$\phi_1(M,r,z) \leq r_1$, a contradiction.

By \textit{highest rank independence}, $\lambda$ is
independent of $r_{2}$ and $z_{2}$. Moreover, $\lambda$ is independent of
$r_{1}$ and $z_{1}$. 
To see this, suppose, by contradiction, that we have
$\tilde{r}=(\tilde{r}_{1},\tilde{r}_{2})$ with $r_{2}=\tilde{r}_{2}$,
$\tilde{z}=(\tilde{z}_{1},\tilde{z}_{2})$ with $z_{2}=\tilde{z}_{2}$
and such that $\phi_{1}(M,r,z)=z_1+\lambda (r_{1}-z_1)$ and
$\phi_{1}(M,\tilde{r},\tilde{z})=\tilde{z}_1+\widetilde{\lambda}(\tilde{r}_{1}-\tilde{z}_{1})$ with
$\lambda\neq\widetilde\lambda$. 
By \textit{bilateral linearity},
$$\phi_{1}(M,r,z)=\phi_{1}(M,r-z+z,0+z)=\phi_{1}(M,r-z,0)+\phi_{1}(M,z,z)=$$ $$(r_1-z_1)\phi_{1}(M,(1,\frac{(r_2-z_2)}{(r_1-z_1)}),0)+\phi_{1}(M,z,z).$$
By \textit{needs lower bound}, $\phi_{1}(M,z,z)=z_1$. Therefore, 
$\phi_{1}(M,r,z)=(r_1-z_1)\phi_{1}(M,(1,\frac{r_2-z_2}{r_1-z_1}),0)+z_{1}.$ 
Then, $\lambda=\phi_{1}(M,(1,\frac{r_2-z_2}{r_1-z_1}),0).$
Again, by \textit{bilateral linearity},
$$\phi_{1}(M,\tilde{r},\tilde{z})=\phi_{1}(M,\tilde{r}-\tilde{z}+\tilde{z},0+\tilde{z})=\phi_{1}(M,\tilde{r}-\tilde{z},0)+\phi_{1}(M,\tilde{z},\tilde{z}).$$
By \textit{needs lower bound}, $\phi_{1}(M,\tilde{z},\tilde{z})=\tilde{z}_1$. Therefore, 
$\phi_{1}(M,\tilde{r},\tilde{z})=(\tilde{r}_1-\tilde{z}_1)\phi_{1}(M,(1,\frac{\tilde{r}_2-\tilde{z}_2}{\tilde{r}_1-\tilde{z}_1}),0)+\tilde{z}_{1}.$
Then, $\tilde{\lambda}=\phi_{1}(M,(1,\frac{\tilde{r}_2-\tilde{z}_2}{\tilde{r}_1-\tilde{z}_1}),0).$
Altogether, $\phi_{1}(M,(1,\frac{r_2-z_2}{r_1-z_1}),0)\neq \phi_{1}(M,(1,\frac{\tilde{r}_2-\tilde{z}_2}{\tilde{r}_1-\tilde{z}_1}),0)$, which contradicts our previous statement that such a value is independent of $\frac{r_2-z_2}{r_1-z_1}$ and $\frac{\tilde{r}_2-\tilde{z}_2}{\tilde{r}_1-\tilde{z}_1}$. Now, by \emph{balance},
$\phi_{2}(M,r,z)=r_{2}+r_{1}-\phi_{1}^{\lambda}(M,r,z)=\phi_{2}^{\lambda}(M,r,z)$.

\item[\textbf{Induction Hypothesis:} $\boldsymbol{M=\{1,\dots, k\}.}$] Next, suppose there is $\lambda$ such that $\phi(M, r, z) =\phi^{\lambda}(M, r, z)$ for each problem $(M, r, z)\in \mathcal{Z}$ with $M=\{1,\dots, k\}$.

\item[\textbf{General Case:} $\boldsymbol{M=\{1,\dots, k+1\}.}$] Now, consider the problem $(M, r, z)\in \mathcal{Z}$, with $M=\{1,\dots, k+1\}$.
By \textit{highest rank  independence}, 
$\phi_{i}(M,r,z)=\phi_{i}(M,(r_1,\ldots,r_k,0),(z_1,\ldots,z_k,0))$, for all $i=1,\ldots,k$.

By \textit{highest rank splitting neutrality}, 
$$\phi_{i}(M,(r_1,\ldots,r_k,0),(z_1,\ldots,z_k,0))=\phi_{i}(\{1,\ldots, k\},(r_1,\ldots,r_k),(z_1,\ldots,z_k)),$$
for all $i=1,\ldots,k-1$.
Then,
$\phi_{i}(M,r,z)=\phi_{i}(\{1,\ldots, k\},(r_1,\ldots,r_k),(z_1,\ldots,z_k)),$
for all $i=1,\ldots,k-1.$
By the induction hypothesis, 
$$\phi_{i}(\{1,\ldots, k\},(r_1,\ldots,r_k),(z_1,\ldots,z_k))=\phi^\lambda_{i}(\{1,\ldots, k\},(r_1,\ldots,r_k),(z_1,\ldots,z_k)),$$
for all $i=1,\ldots,k$.
In particular,
\begin{equation}\label{agent1}
    \phi_{1}(M,r,z)=\phi^\lambda_{1}(\{1,\ldots, k\},(r_1,\ldots,r_k),(z_1,\ldots,z_k))=z_1+\lambda(r_1-z_1).
\end{equation}
By \textit{lowest rank consistency}, 
$$\phi_{i}(M,r,z)=\phi_{i}(\{2,\ldots,k+1\},(r_{2}+r_{1}-\phi_{1}(M,r,z),r_{M\setminus\{1,2\}}),z_{\{2,\ldots,k+1\}}),$$
for all $i=2,\ldots,k+1$.
Then, by (\ref{agent1}),
$$\phi_{i}(M,r,z)=\phi_{i}(\{2,\ldots,k+1\},(r_{2}+(1-\lambda)(r_{1}-z_{1}),r_{M\setminus\{1,2\}}),z_{\{2,\ldots,k+1\}}),$$
for all $i=2,\ldots,k+1$. 
By the induction hypothesis, and \textit{lowest rank consistency},
$$\phi_{i}(M,r,z)=\phi_{i}(\{2,\ldots,k+1\},(r_{2}+(1-\lambda)(r_{1}-z_{1}),r_{M\setminus\{1,2\}}),z_{\{2,\ldots,k+1\}})=$$ $$ \phi^{\lambda}_{i}(\{2,\ldots,k+1\},(r_{2}+(1-\lambda)(r_{1}-z_{1}),r_{M\setminus\{1,2\}}),z_{\{2,\ldots,k+1\}})=\phi^{\lambda}_{i}(M,r,z),$$ 
for all $i=2,\ldots,k+1$.
Finally, by \textit{balance},
$$\phi_{1}(M,r,z)=\sum_{i=1}^{k+1} r_{i}- \sum_{i=1}^{k+1} \phi_{i}(M,r,z)=\sum_{i=1}^{k+1} r_{i}- \sum_{i=1}^{k+1} \phi^{\lambda}_{i}(M,r,z)=\phi_{1}^{\lambda}(M,r,z),$$ which concludes the proof.\qquad\endproof
\end{description}

Theorem \ref{teorema de caracterizacion de la geometrica con r>z} is a generalization of Theorem 1 in Hougaard et al. (2017) to this setting. To obtain the generalization we resorted to the natural axiom in this setting of needs lower bound, as well as to a strengthening scale invariance to linearity for two-agent problems. 

\subsection{The three musketeers}
In this subsection, we characterize three specific need-adjusted geometric rules that arise when imposing three additional axioms, each of which is relevant in different contexts. These rules correspond to the extreme cases $\lambda = 0$ and $\lambda = 1$, where transfers are fully operative or absent, respectively, as well as to the intermediate case $\lambda = \frac{1}{2}$.

The case $\lambda = 0$ models the scenario in which each agent receives his needs, and the highest-ranked agent additionally receives the surplus. Formally,

\noindent \textbf{Need-adjusted full-transfer rule, $\boldsymbol{\phi^0}$:} 
For each $(M,r,z)\in \mathcal{Z},$
\[
\phi^0(M,r,z) = 
\left(z_1, \ldots, z_{n-1}, z_n + \sum_{i\in M} r_i - z_i\right).
\]

The case $\lambda = 1$ models the opposite scenario in which each agent receives his own revenue. This rule is simply the identity rule, as no redistribution takes place and each agent retains his original revenue. 
Formally,

\noindent \textbf{No-transfer rule, $\boldsymbol{\phi^1}$:} 
For each $(M,r,z)\in \mathcal{Z},$
\[
\phi^1(M,r,z) = 
\left(r_1, \ldots, r_{n}\right).
\]

Finally, the intermediate case, obtained when $\lambda = \frac{1}{2}$, is what we dub the \textbf{need-adjusted balanced-transfer rule, $\boldsymbol{\phi^{\frac{1}{2}}}$}. Under this rule, agent $i$ receives his needs $z_i$ plus a geometrically discounted carryover of past surpluses. Specifically, his own surplus $r_i - z_i$ enters with weight $\frac{1}{2}$, while the surplus of agent $i-1$ enters with weight $\frac{1}{2^2}$, that of agent $i-2$ with weight $\frac{1}{2^3}$, and so on. The formal definition is given by equations~\eqref{lambdaz} and \eqref{lambdazm} with $\lambda = \frac{1}{2}$.

The following axioms will allow us to single out each of the three rules described above from the family of need-adjusted geometric rules. 

First, a standard impartiality notion stating that agents with equal revenues and needs receive equal awards. Formally,

\bigskip\noindent
 \textbf{Equal Treatment of Equals}: For each $(M,r,z)\in \mathcal{Z}$, and each pair $i,j\in M$ such that $r_{i} = r_{j}$, and $z_{i} = z_{j}$, we have $\phi_{i}(M,r,z)\ =\ \phi_{j}(M,r,z).$\bigskip

Equal treatment of equals ignores the role of the hierarchy. The next axiom restores it with a natural notion of order preservation (which ignores itself agents' revenues).\footnote{This axiom is reminiscent of the so-called \textit{structural monotonicity} axiom for additive games with a linear permission structure (e.g., van den Brink and Gilles, 1996).} Formally,

\bigskip\noindent
 \textbf{Hierarchical Order Preservation}:
For each $(M,r,z)\in \mathcal{Z}$, and each pair $i,j\in M$, where $i
\geq j,$  $ \phi_{i}(M,r,z)-z_i\ \geq\
\phi_{j}(M,r,z)-z_j.$\bigskip

Finally, the third axiom refers to the canonical situation exemplifying the basic trade-off between location in the hierarchy and production. That is, the case in which the highest-ranked agent is not productive (generates zero revenue and, thus, has zero need). For those settings, one might find appealing to allocate the net revenue (that is, after the productive agent covers his need) equally. A plausible rationale is that, although the lowest-ranked agent is the only productive one, the highest-ranked agent is also necessary for production to take place. Formally,

\bigskip\noindent\textbf{Canonical Bilateral Fairness}: For each $r_1\in\mathbb{R}_{+}$, and each $z_1\in\mathbb{R}_{+}$, such that $r_1\ge z_1$,
$\phi(\{1,2\},(r_1,0),(z_1,0))=(z_1+\frac{r_1-z_1}{2},\frac{r_1-z_1}{2}).$\bigskip

The next result shows that adding each of the previous three axioms to those in Theorem \ref{teorema de caracterizacion de la geometrica con r>z} allows us to characterize each of the three specific rules considered above. As a matter of fact, some of the axioms from Theorem \ref{teorema de caracterizacion de la geometrica con r>z} can be dismissed in each statement as they are not needed for the characterization (albeit, needless to say, each rule satisfies them all). 

\begin{theorem}\label{focal rules}
The following statements hold:
\begin{enumerate}[(i)]
    \item A rule satisfies highest rank independence and equal treatment of equals if and only if it is the no-transfer rule, i.e., $\phi^\lambda$ with $ \lambda=1$.
    \item A rule satisfies highest rank splitting neutrality and hierarchical order preservation if and only if it is the need-adjusted full-transfer rule, i.e., $\phi^\lambda$ with $ \lambda=0$. 
    \item A rule satisfies lowest rank consistency, highest rank independence, highest rank splitting neutrality and canonical bilateral fairness if and only if it is the need-adjusted balanced-transfer rule, i.e., $\phi^\lambda$ with $ \lambda=0.5$.
\end{enumerate}
\end{theorem}

 \noindent\textit{Proof of (i):}
We concentrate on the non-trivial implication of the statement.
Let $\phi$ be a rule satisfying \textit{highest rank independence} and \textit{equal treatment of equals}. Let $(M,r,z)\in
\mathcal{Z}$ be given. For each $i\in M$, let $(M,(r_{-n},r_i),(z_{-n},z_i))\in\mathcal{Z}$. By \textit{highest rank independence}, $$\phi_j (M,(r_{-n},r_i),(z_{-n},z_i))=\phi_j(M,r,z),$$ for each $j=1,\dots, n-1$. By balance, $$\phi_n (M,(r_{-n},r_i),(z_{-n},z_i))=r_i+\sum_{j=1}^{n-1}(r_j-\phi_j(M,r,z)).$$
Now, by \textit{equal treatment of equals}, 
$$
\phi_i (M,(r_{-n},r_i),(z_{-n},z_i))=\phi_n (M,(r_{-n},r_i),(z_{-n},z_i)).
$$
And, therefore,
$$
\phi_i(M,r,z)-r_i=\sum_{j=1}^{n-1}(r_j-\phi_j(M,r,z))=\phi_n(M,r,z)-r_n,
$$
where the last equality follows by balance. 
Thus,
$$
0=\sum_{i\in M}(\phi_i(M,r,z)-r_i)=\sum_{i\in M}(\phi_n(M,r,z)-r_n)=n(\phi_n(M,r,z)-r_n),
$$
from where it follows that $\phi_n(M,r,z)=r_n$ and, therefore, $\phi_i(M,r,z)=r_i$, for each $i\in M$, as desired.\qquad  \endproof

\bigskip

 \noindent\textit{Proof of (ii):}
We concentrate on the non-trivial implication of the statement. 
Conversely, let $\phi$ be a rule satisfying \textit{highest
rank splitting neutrality} and \textit{hierarchical order
preservation}. By contradiction, suppose that there exists a
problem $(M,r,z)\in \mathcal{Z}$ and an agent $i\neq m$, such that $\phi_i(M,r,z) =z_i+
\epsilon $ with $\epsilon> 0$.
Let $(M',r',z')$ be defined by setting $M' = \{1,\dots, m+x\}$, $r'_i = r_i$ and $z'_i=z_i$ for all $i<m$, and $\sum_{j=m}^{m+x}r'_j = r_m$ and $\sum_{j=m}^{m+x}z'_j = z_m$ such that $r'_j\geq z'_j$ for all $j=m,\ldots,m+x$.
By \textit{highest rank splitting neutrality} $\phi_i(M',r',z') = \phi_i(M,r,z)$ for all $i<m$. Now, choose $x
> \frac{\sum_{j=1}^{m+x}\left(\phi_j(M',r',z')-z'_j\right)}{\epsilon}$. By \textit{hierarchical order preservation},
$\phi_j(M',r',z')-z'_j \geq \epsilon$ for all $j=m,\dots, m+x$, which contradicts needs lower bound. \qquad  \endproof
\bigskip

 \noindent\textit{Proof of (iii):}
By Theorem 1, we know that $\phi^{\frac{1}{2}}$ satisfies the first three axioms of the statement. It is
straightforward to see that it also satisfies \textit{canonical bilateral fairness}. Conversely, let $\phi$ be a rule satisfying all the axioms in the statement. The proof follows by mathematical induction:
\begin{description}
    \item[\textbf{Base step:} $\boldsymbol{M=\{1,2\}}.$] Let $r=(r_{1},r_{2})$ and $z=(z_{1},z_{2})$. 
By \textit{highest rank independence}, $\phi_{1}(M,r,z)=\phi_{1}(M,(r_{1},0),(z_{1},0))$. By
\textit{canonical bilateral fairness},
$\phi_1(M,(r_{1},0),(z_{1},0))=z_1+\frac{r_1-z_1}{2}$. Then, by \textit{balance},
$\phi(M,r,z)=(z_1+\frac{r_1-z_1}{2},\frac{r_1-z_1}{2})=\phi^{\frac{1}{2}}(M,r,z)$.

\item[\textbf{Induction Hypothesis:} $\boldsymbol{M=\{1,\dots, k\}.}$] Next, suppose there is $\lambda$ such that $\phi(M, r, z) =\phi^{\frac{1}{2}}(M, r, z)$ for each problem $(M, r, z)\in \mathcal{Z}$ with $M=\{1,\dots, k\}$.

\item[\textbf{General Case:} $\boldsymbol{M=\{1,\dots, k+1\}.}$] Now, consider the problem $(M, r, z)\in \mathcal{Z}$, with $M=\{1,\dots, k+1\}$.
By \textit{highest rank independence}, 
$\phi_{i}(M,r,z)=\phi_{i}(M,(r_1,\ldots,r_k,0),(z_1,\ldots,z_k,0))$, for all $i=1,\ldots,k$.

By \textit{highest rank splitting neutrality}, 
$$\phi_{i}(M,(r_1,\ldots,r_k,0),(z_1,\ldots,z_k,0))=\phi_{i}(\{1,\ldots, k\},(r_1,\ldots,r_k),(z_1,\ldots,z_k)),$$
for all $i=1,\ldots,k-1$.
Then,
$\phi_{i}(M,r,z)=\phi_{i}(\{1,\ldots, k\},(r_1,\ldots,r_k),(z_1,\ldots,z_k)),$
for all $i=1,\ldots,k-1.$
By the induction hypothesis, 
$$\phi_{i}(\{1,\ldots, k\},(r_1,\ldots,r_k),(z_1,\ldots,z_k))=\phi^{\frac{1}{2}}_{i}(\{1,\ldots, k\},(r_1,\ldots,r_k),(z_1,\ldots,z_k)),$$
for all $i=1,\ldots,k$.
In particular,
\begin{equation}\label{agent1bis}
    \phi_{1}(M,r,z)=\phi^{\frac{1}{2}}_{1}(\{1,\ldots, k\},(r_1,\ldots,r_k),(z_1,\ldots,z_k))=z_1+\frac{1}{2}(r_1-z_1).
\end{equation}
By \textit{lowest rank consistency}, 
$$\phi_{i}(M,r,z)=\phi_{i}(\{2,\ldots,k+1\},(r_{2}+r_{1}-\phi_{1}(M,r,z),r_{M\setminus\{1,2\}}),z_{\{2,\ldots,k+1\}}),$$
for all $i=2,\ldots,k+1$.
Then, by (\ref{agent1bis}),
$$\phi_{i}(M,r,z)=\phi_{i}(\{2,\ldots,k+1\},(r_{2}+\frac{1}{2}(r_{1}-z_{1}),r_{M\setminus\{1,2\}}),z_{\{2,\ldots,k+1\}}),$$
for all $i=2,\ldots,k+1$. 
By the induction hypothesis, and \textit{lowest rank consistency},
$$\phi_{i}(M,r,z)=\phi_{i}(\{2,\ldots,k+1\},(r_{2}+\frac{1}{2}(r_{1}-z_{1}),r_{M\setminus\{1,2\}}),z_{\{2,\ldots,k+1\}})=$$ $$ \phi^{\frac{1}{2}}_{i}(\{2,\ldots,k+1\},(r_{2}+\frac{1}{2}(r_{1}-z_{1}),r_{M\setminus\{1,2\}}),z_{\{2,\ldots,k+1\}})=\phi^{\frac{1}{2}}_{i}(M,r,z),$$ 
for all $i=2,\ldots,k+1$.
Finally, by \textit{balance},
$$\phi_{1}(M,r,z)=\sum_{i=1}^{k+1} r_{i}- \sum_{i=1}^{k+1} \phi_{i}(M,r,z)=\sum_{i=1}^{k+1} r_{i}- \sum_{i=1}^{k+1} \phi^{\frac{1}{2}}_{i}(M,r,z)=\phi_{1}^{0.5}(M,r,z),$$ which concludes the proof.\qquad\endproof
\end{description}

We conclude this section mentioning that alternatively to what \textit{canonical bilateral fairness} states, one might consider the axiom stating that those (canonical) two-agent problems are solved in a specific (non-egalitarian) way; say $(z_1+\lambda (r_1-z_1), (1-\lambda)(r_1-z_1))$, for each $\lambda\in[0,1]$. We shall refer to it as $\lambda$-\textit{canonical bilateral fairness}.\footnote{To be consistent with the parlance, when $\lambda=\frac{1}{2}$ we simply write \textit{canonical bilateral fairness} instead of $\frac{1}{2}$-\textit{canonical bilateral fairness}.}
A similar proof to the one provided for statement (iii) of Theorem \ref{focal rules} could be provided to characterize any specific need-adjusted geometric rule (not only the need-adjusted balanced-transfer rule) by replacing canonical bilateral fairness with $\lambda$-canonical bilateral fairness. 

\subsection{And D'Artagnan}

We start this section strengthening two of the axioms considered above: highest rank independence and canonical bilateral fairness. 

The first notion can naturally be extended beyond the top agent to a top group of agents. Formally, 

\bigskip\noindent \textbf{Superior Independence}: For each $(M,r,z)\in \mathcal{Z}$, each $\hat{r}_{\{i,\ldots,m\}}\in \Bbb{R}^{\{i,\ldots,m\}}$ and each $\hat{z}_{\{i,\ldots,m\}}\in \Bbb{R}^{\{i,\ldots,m\}}$ such that $\hat{r}_{j}\ge \hat{z}_{j}$, for each $j\in \{i,\ldots,m\}$, 
\[
\phi_{M\setminus\{i,\ldots,m\}}(M,r,z)=\phi_{M\setminus\{i,\ldots,m\}}\left(M,(r_{-\{i,\ldots,m\}},\hat{r}_{\{i,\ldots,m\}}),(z_{-\{i,\ldots,m\}}, \hat{z}_{\{i,\ldots,m\}})\right).
\]

As for the second notion, one might want to extend it to the case with more than two agents, but in which only the lowest-ranked agent generates revenues. As in that case, all agents in the hierarchy (except for the lowest-ranked agent) generate zero revenue (and, thus, have zero need). Hence, there might be compelling reasons to treat them equally.\footnote{As Ju et al. (2025) argue for a similar axiom in a setting without needs, all superiors can only contribute to the generation of revenue through collaborating with the lowest-ranked agent and, in that sense, all superiors contribute equally in this problem.} And likewise for the (lowest-ranked) productive agent, once his needs are covered. Formally, 

\bigskip\noindent \textbf{Canonical Symmetric Fairness}: For each $(M,r,z)\in \mathcal{Z}$ with $r_k=0=z_k$ for each $k>1$,   
$$
\phi_1(M,r,z)-z_1=\phi_k(M,r,z),
$$
for each $k>1$. 
\bigskip


The following rule (outside the family of geometric rules studied above) satisfies the two axioms just introduced. This rule suggests that each agent keeps his need and shares the residual equally with all his predecessors. The rule is akin to the namesake cost-sharing rule introduced by Moulin and Shenker (1992). Formally,
\bigskip

\noindent\textbf{Need-adjusted serial rule}, $\boldsymbol{\phi^s}$: For each $(M,r,z)\in \mathcal{Z}$ and each $i\in M$,
$$ 
\phi^{s}_i(M,r,z) = z_i+\frac{r_1-z_1}{n}+\frac{r_2-z_2}{n-1}+\frac{r_3-z_3}{n-2}+\ldots+\frac{r_i-z_i}{n-i+1} = z_i+\sum_{j \leq i} \frac{r_j-z_j}{n-j+1}. 
$$

It turns out, as stated in the next result, that this rule is characterized when combining those two axioms with lowest rank consistency. 

\begin{theorem}\label{serial} A rule $\phi$ satisfies {\sl lowest rank consistency}, {\sl superior independence}, and {\sl canonical symmetric fairness} if and only if it is the serial rule.
\end{theorem}

\noindent \textit{Proof:} It is not difficult to show that the need-adjusted serial rule
satisfies all the axioms in the statement of the theorem. Conversely, let $\phi$ be a rule satisfying all the axioms in
the statement of the theorem. 
The proof is by induction.
\begin{description}
    \item[\textbf{Base step:} $\boldsymbol{M=\{1,2\}}.$]
That is, $r=(r_{1},r_{2})$, and $z=(z_{1},z_{2})$. By \textit{superior independence}, $\phi_{1}(M,r,z)=\phi_{1}(M,(r_{1},0),(z_1,0))$. By
\textit{canonical symmetric fairness},
$z_1-\phi_1(M,(r_{1},0),(z_1,0))=\phi_2(M,(r_{1},0),(z_1,0))$.
Then, by \textit{balance},
$\phi(M,r,z)=(z_1+\frac{r_{1}-z_1}{2},r_{2}+\frac{r_{1}-z_1}{2})=\phi^{s}(M,r,z)$.
\item[\textbf{Induction Hypothesis:} $\boldsymbol{M=\{1,\dots, k\}.}$] Next, suppose that $\phi(M, r, z)  \equiv \phi^{s}(M, r, z)$ for each problem $(M, r, z)\in \mathcal{Z}$ with $M=\{1,\dots, k\}$.

\item[\textbf{General Case:} $\boldsymbol{M=\{1,\dots, k+1\}.}$] By
\textit{superior independence}, 
\begin{equation}\label{ecu 1 caracterizacion serial with needs}
    \phi_{1}(M,r,z)=\phi_{1}(M,(r_1,\underbrace{0,\ldots,0}_{k-times}),(z_1,\underbrace{0,\ldots,0}_{k-times})).
\end{equation}

By \textit{canonical symmetric fairness},
\begin{equation}\label{ecu 2 caracterizacion serial with needs}
    \phi_{1}(M,(r_1,\underbrace{0,\ldots,0}_{k-times}),(z_1,\underbrace{0,\ldots,0}_{k-times}))=z_1+\frac{r_1-z_1}{k+1}=\phi^s_{1}(M,r,z).
\end{equation}

By \emph{lowest rank consistency},
$$\phi_{i}(M,r,z)=\phi_{i}(M\setminus \{1\},r_{2}+r_{1}-\phi_{1}(M,r,z),r_{M\setminus\{1,2\} },z_{M\setminus \{1\}}),$$ for all $i=2,\ldots,k+1.$

Then, by \eqref{ecu 1 caracterizacion serial with needs} and \eqref{ecu 2 caracterizacion serial with needs}, the induction hypothesis, and \textit{lowest rank consistency}
$$\phi_{i}(M,r,z)=\phi_{i}(M\setminus \{1\},r_{2}+r_{1}-\phi^s_{1}(M,r,z),r_{M\setminus\{1,2\} },z_{M\setminus \{1\}})$$
$$= \phi^{s}_{i}(M\setminus\{1\},r_{2}+r_{1}-\phi^s_{1}(M,r,z),r_{M\setminus\{1,2\}},z_{M\setminus\{1\}})=\phi^s_{i}(M,r,z),$$
 for all $i=2,\ldots,k+1,$ which concludes the proof.\qquad\endproof
\end{description}

Note that if we compare statement (iii) of Theorem \ref{focal rules} and Theorem \ref{serial}, we can see that canonical bilateral fairness and highest rank independence in the former are strengthened, respectively, to canonical symmetric fairness and superior independence in the latter. Lowest rank consistency is used in both results and highest rank splitting neutrality in statement (iii) of Theorem \ref{focal rules} is dismissed for Theorem \ref{serial}. 

\section{Further insights}
\subsection{When needs can be superior to revenues}
A key assumption in the previous analysis was that each individual's needs were below the corresponding own revenues. 
We now relax this requirement and impose only that the group as a whole satisfies the inequality; that is, total revenues exceed total needs. 
Formally, we assume
\[
\sum_{i\in M} r_i \geq \sum_{i\in M} z_i,
\]
but not necessarily that \( r_i \geq z_i \) for each \( i \in N \).
We denote by \( \mathcal{Z}^\star \) the domain of problems satisfying this condition.
The axioms and rules introduced introduced above can be extended in a natural way to this domain 
\( \mathcal{Z}^\star \). 
We now explore how the above axiomatic analysis applies in this broader setting.

We first see that, somewhat surprisingly, the counterpart of Theorem~1 in this setting identifies the need-adjusted full-transfer rule as the unique rule satisfying \emph{highest rank independence} and \emph{needs lower bound}. Nevertheless, this rule also satisfies \emph{superior independence}, \emph{lowest rank consistency}, \emph{highest rank splitting neutrality}, and \emph{bilateral linearity}. 
Formally,

\begin{proposition}
On $\mathcal{Z}^\star$, a rule satisfies highest rank independence and needs lower bound if and only if it is the need-adjusted full-transfer rule, i.e., $\phi^\lambda$ with $\lambda=1$. 
\end{proposition}
\begin{proof}
    It is straightforward to show that the need-adjusted full-transfer rule satisfies \emph{highest rank independence} and \emph{needs lower bound} on $\mathcal{Z}^\star$. Conversely, suppose that a rule $\phi$ satisfies the two axioms on $\mathcal{Z}^\star$. By contradiction, assume that $\phi$ is not the need-adjusted full-transfer rule. That is, there exists $(M,r,z)\in \mathcal{Z}^\star$ and $i<m$ such that $\phi_i(M,r,z) > z_i$.  

    By \emph{needs lower bound} (and \emph{balance}),
\[
\sum_{j \in M} z_j < \sum_{j \in M} \phi_j(M,r,z) = \sum_{j \in M} r_j.
\]
Let $z'$ be such that $z'_1 > z_1$, $z'_j = z_j$ for each $j \in M \setminus \{1\}$, and 
$\sum_{j \in M} z'_j = \sum_{j \in M} r_j$. Thus, $(M,r,z')\in \mathcal{Z}^\star$. 
By \emph{highest rank independence},
\[
\phi_i(M,r,z') = \phi_i(M,r,z) > z_i = z'_i.
\]
Moreover, by \emph{needs lower bound},
\[
\phi_j(M,r,z') \geq z'_j,
\]
for each $j \in M.$ 
Therefore,
\[
\sum_{j \in M} \phi_j(M,r,z') > \sum_{j \in M} z'_j = \sum_{j \in M} r_j,
\]
which is a contradiction.    
\end{proof}
\bigskip

The previous result shows that, within \( \mathcal{Z}^\star \), the needs lower bound axiom becomes rather demanding. 
Following Mart\'{\i}nez and Moreno-Ternero (2024), we instead consider a weaker version of this axiom, which requires the condition to hold only for problems satisfying \( r \geq z \). 
This relaxed requirement is fulfilled by the natural extensions of geometric rules to \( \mathcal{Z}^\star \), allowing us to recover the full characterization of geometric rules on this domain. Formally,

\bigskip\noindent \textbf{Weak Needs Lower Bound}: For each $(M,r,z)\in \mathcal{Z}^\star$ with $r\geq z$, and each $i\in M$,
\[
\phi_{i}(M,r,z)\ge z_i.
\]

\begin{proposition}
    On $\mathcal{Z}^\star$, a rule satisfies weak needs lower bound, lowest rank consistency, highest rank independence, highest rank splitting neutrality, and bilateral linearity if and only if it is a need-adjusted geometric rule.
\end{proposition}
\textit{Proof:} As in Theorem \ref{teorema de caracterizacion de la geometrica con r>z}, it is not difficult to see that need-adjusted geometric rules
satisfy all the axioms in the statement of the theorem. 
Conversely, let $\phi$ be a rule on $\mathcal{Z}^\star$, satisfying all the axioms in
the statement of the theorem. 
The proof is by induction:
\begin{description}
    \item[\textbf{Base step:} $\boldsymbol{M=\{1,2\}}.$] Let $r=(r_{1},r_{2})$ and $z=(z_{1},z_{2})$. 

Assume first that $r \geq z$. By the proof of Theorem~\ref{teorema de caracterizacion de la geometrica con r>z}, there exists a parameter $\lambda$ such that, for every $(M,r,z)$ with $r \geq z$,  
\begin{equation}\label{ecu 1 teorema geometrica z>r}
    \phi(M,r,z) = \phi^{\lambda}(M,r,r).
\end{equation}

Next, consider a problem $(M,r,z)\in \mathcal{Z}^{\star}$ where $z_2 > r_2$. Hence, $r_1 > z_1$.  
By \textit{highest rank independence}, we have  
\[
\phi_1(M,r,z) = \phi_1(M,(r_1,0),(z_1,0)).
\]
As $(r_1,0) \geq (z_1,0)$, by~\eqref{ecu 1 teorema geometrica z>r}, we obtain  
\[
\phi_1(M,(r_1,0),(z_1,0)) = \phi^{\lambda}_1(M,(r_1,0),(z_1,0)).
\]
Therefore, $\phi_1(M,r,z) = \phi^{\lambda}_1(M,r,z)$.  
As $|M|=2$, it follows that $\phi(M,r,z) = \phi^{\lambda}(M,r,z)$.

Finally, consider the case where $z_1 > r_1$. Then, $r_2 > z_2$.  
As 
\[
\left((z_1,r_2),(z_1,z_2)\right) = (r,z) + \left((z_1 - r_1,0),(0,0)\right),
\]
it follows by \textsl{bilateral linearity} that
\[
\phi\left(M,(z_1,r_2),(z_1,z_2)\right) = \phi(M,r,z) + \phi\left(M,(z_1 - r_1,0),(0,0)\right).
\]
As $(z_1,r_2) \geq (z_1,z_2)$ and $(z_1 - r_1,0) \geq (0,0)$, by~\eqref{ecu 1 teorema geometrica z>r} we obtain  
$$
\phi\left(M,(z_1,r_2),(z_1,z_2)\right) = \phi^{\lambda}\left(M,(z_1,r_2),(z_1,z_2)\right),
$$
and
$$
\phi\left(M,(z_1 - r_1,0),(0,0)\right) = \phi^{\lambda}\left(M,(z_1 - r_1,0),(0,0)\right).
$$
Therefore,
\[
\phi(M,r,z) = \phi^{\lambda}\left(M,(z_1,r_2),(z_1,z_2)\right) - \phi^{\lambda}\left(M,(z_1 - r_1,0),(0,0)\right)
= \phi^{\lambda}(M,r,z).
\]

The rest of the proof is analogous to that of Theorem~\ref{teorema de caracterizacion de la geometrica con r>z}.
\endproof
\end{description}

The characterizations of the need-adjusted balance-transfer rule and the no-transfer rule from Theorem 2 extend to \( \mathcal{Z}^\star \) without any modification. Likewise, as the extension of the need-adjusted serial rule to $\mathcal{Z}^\star$ also satisfies weak needs lower bound (but not needs lower bound), Theorem \ref{serial} extends to $\mathcal{Z}^\star$ without any modification. 

\begin{proposition}\label{serial cuando z > r} The following statements hold:
\begin{enumerate}
    \item On $\mathcal{Z}^\star$, a rule satisfies highest rank independence and equal treatment of equals if and only if it is the no-transfer rule, i.e., $\phi^\lambda$ with $ \lambda=1$.
    \item On $\mathcal{Z}^\star$, a rule satisfies {\sl lowest rank consistency}, {\sl highest rank independence}, {\sl highest rank splitting neutrality}, 
and {\sl canonical bilateral fairness} if and only if it is the need-adjusted balanced-transfer rule, i.e., $\phi^\lambda$ with $ \lambda=0.5$.
    \item On $\mathcal{Z}^\star$, a rule satisfies lowest rank consistency, superior independence and canonical symmetric fairness if and only if it is the serial rule.
\end{enumerate}
\end{proposition}

\subsection{When needs are all null}
While in the previous section we considered a larger domain of problems (in which needs could be larger than revenues individually, but not collectively), in this one we consider the subdomain of problems in which all agents have zero needs. That is equivalent to the model dealing with (linear) hierarchy (revenue sharing) problems introduced by Hougaard et al. (2017). To fix notation, we shall denote this subdomain as $\mathcal{Z}^{0}$. 
Let $\{\varphi^{\lambda}\}_{\lambda\in[0,1]}$ be the family of \textbf{geometric rules} defined for each $(M,r)\in \mathcal{Z}^{0}$ as follows:
$$\varphi_i^{\lambda} (M,r) = \lambda \left(r_i + (1- \lambda)r_{i-1} + \dots + (1-\lambda)^{i-1}r_1\right),$$ for $i= 1,\dots,
m-1$ and
$$\varphi_m^{\lambda}(M,r) = r_m + (1- \lambda)r_{m-1} + \dots + (1-\lambda)^{m-1}r_1.$$
Our Theorem 1 converts into Theorem 1 in Hougaard et al. (2017), as bilateral linearity is no longer necessary. Likewise, for Theorem 2. As for Theorem 3, we have that a rule (defined over $\mathcal{Z}^{0}$) satisfies lowest rank consistency, superior independence and canonical symmetric fairness if and only if it is the following rule, studied by Ju et al. (2025):

\bigskip\noindent\textbf{Serial rule}. For each $(M,r)\in \mathcal{Z}^{0}$ and each $i\in M$,
$$ 
\varphi^{s}_i(M,r) = \frac{r_1}{n}+\frac{r_2}{n-1}+\frac{r_3}{n-2}+\ldots+\frac{r_i}{n-i+1} = \sum_{j \leq i} \frac{r_j}{n-j+1}. 
$$

Based on the above, one could naturally think of a (linear) hierarchy with needs (revenue sharing) problem as an extension of a (linear) hierarchy (revenue sharing) problem, upon amending the latter with individual needs. That is, one might consider an \textit{extension operator} (e.g., Hougaard et al., 2013) to convert rules solving standard problems into rules that can also solve problems with individual needs. 
A natural way to describe an extension operator would be via a two-stage process in which we first assign agents their needs, and then allocate the resulting (aggregate) surplus, using a standard rule for the standard (linear) hierarchy revenue sharing problem, after reducing individual revenues by needs. 

By definition, this extension operator yields (extended) rules that satisfy the needs lower bound axiom. 
Furthermore, it is straightforward to see that the family of need-adjusted geometric rules $\{\phi^{\lambda}\}_{\lambda\in[0,1]}$ characterized above is the image via this operator of the geometric rules $\{\varphi^{\lambda}\}_{\lambda\in[0,1]}$ introduced by Hougaard et al. (2017) via the extension operator. 
Likewise, the need-adjusted serial rule $\phi^{s}$ characterized above is the image of the serial rule $\varphi^{s}$ studied by Ju et al. (2025) via this operator. 


Conversely, suppose $\phi$ is an extended rule (i.e., defined over the domain $\mathcal{Z}$) that satisfies \textit{needs lower bound} and the following axiom. 

\bigskip\noindent
 \textbf{Decomposability}: For each $(M,r,z)\in \mathcal{Z}$,
 \[
\phi(M, r, z)\ = \phi(M,z,z)+\phi(M,r-z,0).
\]
Then, 
$\phi$ is the image via the above extension operator of the same rule, when defined for the basic domain of problems without needs $\mathcal{Z}^{0}$. 

\subsection{When only two agents exist}
Finally, we comment further on the two-agent case, which plays a crucial role in our analysis. On the one hand, 
most of the proofs of our results are via induction, which requires to address such a base case in the first place. 
On the other hand, and perhaps more importantly, such a base case captures the fundamental tension of this model, between location and revenues (with our without needs). Such a tension is precisely captured with the axiom of canonical bilateral fairness, or its generalizations (the $\lambda$-canonical bilateral fairness).\footnote{Note that, with highest rank independence, the two-agent case can be reduced to one in which the boss yields zero (revenue and, hence, also need) and thus the issue is to decide how much the productive subordinate should transfer up to the unproductive boss.} 

It turns out that, in the two-agent case, both the need-adjusted balanced-transfer rule, $\phi^{\frac{1}{2}}$ and the need-adjusted serial rule $\phi^{s}$ coincide. Formally, if we denote by $\mathcal{Z}^2$ the domain of two-agent problems (with the notational convention of $M=\{1,2\}$), then we have that
for each $(M,r,z)\in \mathcal{Z}^2$, 
$$\phi^{\frac{1}{2}}(M,r,z)=(z_1+\frac{r_1-z_1}{2},\frac{r_1-z_1}{2})=\phi^{s}(M,r,z).$$
To accommodate both, we shall refer to such a solution in the two-agent case as the \textit{folk solution}. 
We have the following result, which is immediately derived from the proofs of Theorems 1 or 3. 

\begin{corollary}
On $\mathcal{Z}^2$, a rule satisfies highest rank independence and canonical bilateral fairness if and only if it is the folk solution.
\end{corollary}

It follows from the above that the folk solution admits multiple extensions to the general case of more than two agents. The need-adjusted balanced-transfer rule, $\phi^{\frac{1}{2}}$, and the need-adjusted serial rule $\phi^{s}$, are two of them. Although both rules rely on the axiom of lowest rank consistency, the former relies too on the variable-population axiom of highest rank splitting neutrality, whereas the latter relies on strengthening the two axioms from Corollary 1 to superior independence and canonical symmetric fairness. 

Now, canonical symmetric fairness (or null superiors symmetry, as in Ju et al., 2025) might be considered a too strong axiom, as it ignores the role of the hierarchy in the allocation process. In other words, one might have compelling reasons to treat different positions in the hierarchy differently (even if they all yield zero revenue, and thus have zero need). Alternative axioms (to canonical symmetric fairness) formalizing this feature would provide different extensions of the folk solution. 

Similarly, alternative consistency axioms (in which the leftover revenue from the departing agent is not necessarily inherited by the immediate predecessor) would generate other extensions of the folk solution. 
Related to that, variations of Corollary 1 in which canonical bilateral fairness is replaced by $\lambda$-canonical bilateral fairness, provide characterizations for all need-adjusted geometric rules in the two-agent case. We know from Theorem 1 that resorting to lowest rank consistency (and highest rank splitting neutrality) we extend these rules to the general case of more than two agents, in the way we defined the family of need-adjusted geometric rules. But alternative axioms of consistency would give rise to alternative rules. 

Finally, the ground that the need-adjusted serial rule shares with need-adjusted geometric rules (especially, its intermediate member) goes beyond what Corollary 1 says. More precisely, the need-adjusted serial rule satisfies all the axioms in Theorem 1, except for highest rank splitting neutrality. Thus, dismissing such an axiom from the statement encompasses both the need-adjusted serial rule, as well the whole family of need-adjusted geometric rules. It remains an open question to close the whole set of rules satisfying those axioms. That is, to characterize the rules that satisfy lower bound, lowest rank consistency, highest rank independence, and bilateral linearity.

\end{document}